\documentclass{article}
\usepackage{arxiv}

\final

\newtheorem{theorem}{Theorem}
 
\theoremstyle{definition}

\begin{document}

\title{How to Cut a Cake Fairly: \\ A Generalization to Groups}
\markright{Notes}
\author{Erel Segal-Halevi and Warut Suksompong}

\maketitle

\begin{abstract}
A fundamental result in cake cutting states that for any number of players with arbitrary preferences over a cake, there exists a division of the cake such that every player receives a single contiguous piece and no player is left envious.
We generalize this result by showing that it is possible to partition the players into groups of any desired sizes and divide the cake among the groups so that each group receives a single contiguous piece, and no player finds the piece of another group better than that of the player's own group.
\end{abstract}

\section{Introduction.}

In a small town, there is a public basketball court in which ten players (two teams of five players each) can play at a time. 
On a certain Sunday, 30 players want to play on the court, and they have different preferences regarding the time of day at which they prefer to play.
The court manager therefore needs to partition the players into three groups of ten players each, 
and divide the time of day into three contiguous intervals---one interval per group---so that each group of ten can play in its designated time slot.
To avoid complaints, the manager would like the partition to be \emph{envy-free}: every player should believe that the time slot in which he or she plays is at least as good as each of the time slots given to the other two groups. 

In this note, we show that no matter what the players' preferences on the time of day are, there always exists a partition of the players together with an envy-free division of the time into contiguous time slots.

\section{Cake Cutting.}

Before we describe our setting formally, let us take a step back and ask a simpler question: Instead of a group activity like basketball, what if we are trying to divide the time for an individual activity, say running on a treadmill?
In this case, there is no need to partition the players, since only one player will use the treadmill at a time.
The challenge thus lies in dividing the time and assigning the resulting time slots to the players in such a way that no player envies another.

It turns out that this simpler setting falls precisely under the well-studied framework of \emph{cake cutting}.\footnote{Several excellent surveys on the subject are available, e.g., \cite{BramsTa96,RobertsonWe98}.}
The cake serves as a metaphor for resources such as time or land, and the aim is to divide the cake between players in a fair manner.
Formally, there are $n$ players, and the cake is represented by an interval of length $1$.
Each (contiguous) $n$-partition of the cake can be defined by an $n$-tuple $(x_1,x_2,\dots,x_n)$ with $\sum_{i=1}^n x_i = 1$, where $x_i\geq 0$ denotes the length of the $i$th piece from the left.
Given a partition, we say that a player \emph{prefers} a certain piece if the player thinks there is no better piece in the partition; note that this preference can depend on the entire partition.
Each player prefers at least one piece in any given partition, and may prefer more than one piece in case of ties.

The following two (weak) assumptions on the players' preferences are standard in the literature~\cite{Stromquist80,Su99}:
\begin{enumerate}
\item \emph{Hungry players.} Players never prefer an empty piece.
\item \emph{Closed preference sets.} Any piece that is preferred for a convergent sequence of partitions is also preferred at the limiting partition.
\end{enumerate}

A seminal result in cake cutting is that, as long as these two assumptions are satisfied, an envy-free division of the cake always exists. 
It was proved some decades ago in this \textsc{Monthly} by Stromquist~\cite{Stromquist80} via topological arguments and Su~\cite{Su99} using Sperner's lemma, and generalizes earlier results by Dubins and Spanier~\cite{DubinsSp61}, also in this \textsc{Monthly}, and Woodall~\cite{Woodall1980Dividing}. 
\begin{theorem}[\cite{Stromquist80,Su99}]
\label{thm:stromquist}
Let $n$ be any positive integer.
For $n$ hungry players with closed preference sets, there is always a partition of the cake into $n$ contiguous pieces such that each player prefers a different piece.
\end{theorem}

Theorem~\ref{thm:stromquist} is an existence result, so the reader may naturally wonder how one could compute a desired solution.
When there are two players, a contiguous envy-free division can be found using the so-called ``cut-and-choose'' protocol: the first player cuts the cake into two pieces so that she prefers both of them (this is always possible due to our two assumptions), and the second player chooses a piece that she prefers.
For any number of players, however, the problem becomes surprisingly difficult---for example, with three players, Stromquist~\cite{Stromquist08} showed that no finite algorithm can always compute such a division.
We refer the reader to \cite{GoldbergHoSu20} and the references therein for more details.

\section{Our Result.}

Let us now return to our basketball court example and introduce a generalization of the cake-cutting model that captures it.
Again, there are $n$ players who have preferences over the cake.
However, unlike in canonical cake cutting, we are also given positive integers $k_1,k_2,\dots,k_m$ whose sum is $n$, and the players should be divided into $m$ groups with group $j$ containing $k_j$ players.
The cake is then partitioned into $m$ pieces, and group $j$ is assigned the $j$th piece from the left.
Our example thus corresponds to the case where $n=30$, $m=3$, and $k_1=k_2=k_3=10$.
In general, the group sizes may be different---for example, if only $28$ players come to the court, the manager may decide to split them into two groups of ten and one group of eight.

A first idea that comes to mind for proving the existence of an envy-free allocation is to apply Theorem~\ref{thm:stromquist} to the individual players, and then group them according to their time slots.
Specifically, given an envy-free allocation to the individuals such that player $i$ gets piece $i$ in the $n$-partition $(x_1,x_2,\ldots,x_n)$,
we construct groups $G_1=\{1,2,\ldots,k_1\}$,
$G_2=\{k_1+1,k_1+2,\ldots,k_1+k_2\}$, and so on.
We then take the $m$-partition $(y_1,y_2,\ldots,y_m)$
where $y_j = \sum_{i\in G_j}x_i$,
and allocate the $j$th piece from the left to group $G_j$.
Unfortunately, this idea does not work---the resulting allocation might not be envy-free. 
Indeed, going back to our basketball example, it is possible that player 7 prefers piece 7 to all other 29 pieces in the partition $(x_1,x_2,\ldots,x_{30})$, but does not prefer the union of pieces 1, 2, $\dots$, 10 to the union of pieces 11, 12, $\dots$, 20 in the partition $(y_1,y_2,y_3)$.

Nevertheless, we establish below that an envy-free allocation is guaranteed to exist in this general setting.

\begin{theorem}
\label{thm:general}
Let $n\geq m$ be any positive integers, and let $k_1,k_2,\dots,k_m$ be positive integers such that $\sum_{j=1}^m k_j = n$.
For $n$ hungry players with closed preference sets, there is always a partition of the cake into $m$ contiguous pieces, along with a division of the players into $m$ groups with group $j$ containing $k_j$ players, such that each player in group $j$ prefers the $j$th piece from the left in the partition.
\end{theorem}

Note that Theorem~\ref{thm:general} is a strict generalization of Theorem~\ref{thm:stromquist}: taking $n=m$ and $k_j=1$ for all $j$ reduces the former theorem to the latter.
Before we prove Theorem~\ref{thm:general}, a few remarks are in order.

If the partition of the players into groups is fixed in advance (in contrast to our setting), then an envy-free allocation is no longer guaranteed to exist. For example, suppose there are two groups, each of which contains a player who prefers to play in the morning (specifically, only prefers the piece containing the interval $[0,0.1]$ if such a piece exists in the partition) and a player who prefers to play in the evening (specifically, only prefers the piece containing the interval $[0.9,1]$ if such a piece exists in the partition). Regardless of how the interval $[0,1]$ is partitioned and assigned to the two groups, it is not hard to see that at least one player will be envious.

In light of this negative result, relaxations have been proposed for the fixed-group setting, for instance allowing more than one connected piece per group~\cite{SegalhaleviNi19} or ensuring fairness only for a certain fraction of the players in each group~\cite{bade2018fair,SegalhaleviSu19}. Both of these solutions are far from ideal in our basketball example.
Hence, the possibility of constructing ad hoc groups plays a crucial role in ensuring that each group receives a contiguous piece and all players are nonenvious.\footnote{Envy-freeness with variable groups as in our model has recently been studied
for different problems: allocating indivisible items between two groups \cite{KyropoulouSuVo19},
and dividing rooms and rent \cite{AzrieliSh14,GhodsiLaMo18}.
For fixed groups, fair division has been investigated using a probabilistic approach \cite{ManurangsiSu17} and with respect to other fairness notions \cite{Suksompong18}; see \cite{Suksompong18-2} for an overview.
}

Theorem~\ref{thm:general} can be shown by applying a generalized Sperner-type lemma similar to the one proved recently by Meunier and Su~\cite{meunier2019multilabeled}. 
However, here we present a simpler proof using Theorem~\ref{thm:stromquist}.

\begin{proof}[Proof of Theorem~\ref{thm:general}]
We first introduce some notation.
For every positive integer $t$, denote by $[t]$ the set $\{1,2,\ldots,t\}$.
For each $j\in[m]$,
denote by $K_j$ the sum of sizes of the first $j$ groups, i.e., $K_j := \sum_{j'=1}^j k_{j'}$, and let $K_0 := 0$.
Let $G_j := \{K_{j-1}+1,K_{j-1}+2,\ldots,K_j\}$.
Note that $|G_j| = k_j$ and $\cup_{j\in[m]}G_j = [n]$.

We are given $n$ hungry players. The preferences of each player $i\in[n]$ are represented by a \emph{demand function} $f_i$, which assigns to each $m$-partition a nonempty subset of $[m]$ representing the indices of the pieces that player $i$ prefers in that partition.

For each player $i\in[n]$, we construct a new demand function $g_i$, which
assigns to each $n$-partition $\mathbf{x} = (x_1,\ldots,x_n)$ a nonempty subset of $[n]$.
The function $g_i$ is constructed as follows.
\begin{itemize}
\item Given an $n$-partition $\mathbf{x}$, construct an $m$-partition $\mathbf{\widehat{x}}$ by uniting the $k_1$ leftmost intervals into a single interval, the next $k_2$ intervals into another single interval, and so on. Formally, for each $j\in[m]$, let $\widehat{x}_j := \sum_{i\in G_j} x_i$.
\item Find $f_i(\mathbf{\widehat{x}})$ --- the preferred piece(s) of player $i$ in the $m$-partition $\mathbf{\widehat{x}}$.
\item For every $j\in f_i(\mathbf{\widehat{x}})$, 
let $M_j := \arg\max_{i\in G_j}x_i$. That is, $M_j$ contains the indices of the interval(s) in $\mathbf{x}$ that are longest among the subintervals of interval $j$ in $\mathbf{\widehat{x}}$.
\item Let $g_i(\mathbf{x}) := \cup_{j\in f_i(\mathbf{\widehat{x}})} M_j$.
\end{itemize}
Going back to our basketball example, suppose that $f_i(\mathbf{\widehat{x}})=\{2\}$, so in partition~$\mathbf{\widehat{x}}$ player~$i$ prefers to play in the middle time slot.
Suppose that in partition $\mathbf{x}$ the longest among the pieces $11,12,\ldots,20$ are pieces $14$ and $17$. Then $g_i(\mathbf{x}) = \{14,17\}$.

Next, we show that $g_i$ satisfies the two assumptions that are necessary in order to apply Theorem~\ref{thm:stromquist}.
Since the function $f_i$ represents a hungry player,
$f_i(\mathbf{\widehat{x}})$ does not contain any empty piece. Hence, $g_i(\mathbf{x})$, too, does not contain any empty piece, so $g_i$ also represents a hungry player.

We now claim that $g_i$ satisfies the closed preference sets assumption.
To see this, consider a sequence $(\mathbf{x}^d)_{d\in \mathbb{N}}$ of $n$-partitions such that for all partitions in the sequence, $g_i(\mathbf{x}^d)$ contains some fixed $k\in[n]$.
Suppose that the sequence converges to a partition $\mathbf{x}^{\infty}$. To establish the claim, it suffices to show that $k\in g_i(\mathbf{x}^{\infty})$.

For each $d\in\mathbb{N}$, let $\mathbf{\widehat{x}}^d$ be the $m$-partition derived from $\mathbf{x}^d$ by uniting adjacent intervals as in the construction of $g_i$.
The sequence $(\mathbf{\widehat{x}}^d)_{d\in \mathbb{N}}$ converges to some $m$-partition $\mathbf{\widehat{x}}^{\infty}$.
The assumption that $k\in g_i(\mathbf{x}^d)$ implies that for all $d\in\mathbb{N}$:
\begin{enumerate}
\item[(1)] $j\in f_i(\mathbf{\widehat{x}}^d)$, where $j$ is the  unique index in $[m]$ for which $k\in G_j$;
\item[(2)] $x^d_k \geq x^d_{k'}$ for all pieces $k'\in G_j$.
\end{enumerate}
From (1) and since $f_i$ satisfies the closed preference sets assumption, we have $j\in f_i(\mathbf{\widehat{x}}^{\infty})$.
In addition, it follows from (2) that 
$x^{\infty}_k \geq x^{\infty}_{k'}$ for all $k'\in G_j$.
The construction of $g_i$ now implies that $k\in g_i(\mathbf{x}^{\infty})$, as claimed.

We have shown that the preferences represented by $(g_i)_{i\in[n]}$ satisfy the two requirements of Theorem \ref{thm:stromquist}. 
Therefore, there exists an $n$-partition $\mathbf{y}$ in which each player $i\in[n]$ ``prefers'' a different piece according to $g_i$. 
Let $\pi:[n]\rightarrow [n]$ be a bijection such that player $i$ prefers piece $\pi(i)$. In particular, we have $\pi(i)\in g_i(\mathbf{y})$.

Finally, let $\mathbf{\widehat{y}}$ be the $m$-partition derived from $\mathbf{y}$ by uniting adjacent intervals as in the construction of $\mathbf{\widehat{x}}$. 
For each $j\in [m]$, allocate the $j$th piece in $\mathbf{\widehat{y}}$ to group $\pi^{-1}(G_j)$. 
For every player $i\in \pi^{-1}(G_j)$,
since $\pi(i)\in g_i(\mathbf{y})$, 
the construction of $g_i$
implies that 
$j\in f_i(\mathbf{\widehat{y}})$, that is, player $i$ prefers piece $j$ in $\mathbf{\widehat{y}}$.
Hence the constructed $m$-allocation $\mathbf{\widehat{y}}$ is envy-free for the original players.
\end{proof}

\section{Concluding Remarks.}
Theorem \ref{thm:general}
allows 
the players' preferences to depend on the group size. For instance, suppose $m=2$, $k_1=8$, and $k_2=10$. Then, given a partition, each player may decide whether he or she prefers to play in the earlier time slot in a group of eight, or in the later time slot in a group of ten.\footnote{A related problem is the \emph{group activity selection problem}, where each player also has a preference over certain activities and the number of players involved in each activity~\cite{DarmannElKu12}. However, in our setting a player's preference depends in addition on the ``resource'' to which his or her activity is allocated.}
The resulting allocation respects these preferences as long as they satisfy the closed preference sets assumption and the hungry players assumption---in particular, if the earlier piece is empty, then every player must prefer the later piece, and vice versa, regardless of the group sizes.

However, the theorem does \emph{not} allow the players' preferences to depend on the identity of the other players assigned to the same group.
In fact, with this extension, the existence guarantee ceases to hold.
For instance, if there is a popular player, say Alice, such that all other players prefer to play with Alice than not, then the players who are assigned to a different group than hers will necessarily be envious.

For $m\geq 3$ groups, the result of Stromquist~\cite{Stromquist08} implies that no finite protocol can find a contiguous envy-free division even in the simplest case in which each group consists of a single player.
% This impossibility holds even if each player's preferences are represented by a cardinal utility function that is \emph{additive}, meaning that the player's utility for a collection of disjoint pieces is equal to the sum of her utilities for the individual pieces.
However, for $m=2$ groups (and $n = k_1+k_2$ players), 
a simple protocol finds a contiguous envy-free division, as long as the players' preferences satisfy a third condition (in addition to ``hungry players'' and ``closed preference sets''):
\begin{enumerate}
    \item[3.] \emph{Monotone preferences.} Let $P_1$ and $P_2$ be partitions such that for some $\ell\in [m]$, the $\ell$th piece of $P_1$ is contained in the $\ell$th piece of $P_2$, while for every $j\in [m]$ such that $j\neq \ell$, the $j$th piece of $P_2$ is contained in the $j$th piece of $P_1$. Then, if a player prefers the $\ell$th piece of $P_1$, he or she also prefers the $\ell$th piece of $P_2$.
\end{enumerate}
Note that Stromquist's impossibility result for three groups is valid even with monotone preferences. The protocol for two groups works as follows:
\begin{itemize}
    \item Let each player $i\in[n]$ mark a point $x_i\in[0,1]$ so that the player prefers both $[0,x_i]$ and $[x_i,1]$. Such a point must exist by the hungry players and closed preference sets assumptions.
    \item Order the marks from left to right, breaking ties arbitrarily. 
    \item Cut the cake at some point $y$ between the $k_1$th mark and the $(k_1+1)$st mark from the left.
    \item Assign the piece $[0,y]$ to the players who made the $k_1$ leftmost marks, and the piece $[y,1]$ to the remaining $k_2$ players. 
    The monotone preferences assumption guarantees that all agents prefer their allocated piece.
\end{itemize}
The problem of computing an envy-free allocation without the monotone preferences assumption remains for future work.

%For our example, we can move a knife over the ``time cake'' from left to right, and stop at the first moment when at least eight players prefer the left piece.

\begin{affil}
Department of Computer Science, Ariel University, Israel \\
erelsgl@gmail.com
\end{affil}

\begin{affil}
Department of Computer Science, University of Oxford, UK \\
warut.suksompong@cs.ox.ac.uk
\end{affil}

\end{document}